\documentclass{llncs}
\usepackage{etex}
\usepackage[utf8]{inputenc}
\usepackage[T1]{fontenc}
\usepackage{microtype}
\usepackage{mathptmx}
\DeclareMathAlphabet{\mathcal}{OMS}{cmsy}{m}{n}
\usepackage{caption}

\usepackage{color}

\newcommand{\omitteddisc}{}

\newcommand{\pageonenum}{}

\title{The Possibility Problem for Probabilistic XML\\(Extended Version)}
\renewcommand{\omitteddisc}{
  
  This paper is the complete version (including
  proofs) of work initially submitted as an extended abstract (without proofs) at
the AMW~2014 workshop~\cite{amarilli2014possibility} and subsequently submitted (with proofs) at the
BDA~2014 conference (no formal proceedings). This version integrates the
feedback from both rounds of reviews.}
\renewcommand{\pageonenum}{\thispagestyle{plain}}
\let\oldproof\proof
\let\oldendproof\endproof
\renewenvironment{proof}[1][\proofname]{
  \oldproof%
}{\qed \oldendproof}

\usepackage{amssymb, amsmath,  graphicx}
\usepackage{url}
\usepackage{xspace}
\usepackage{paralist}
\usepackage{graphicx}
\usepackage[hidelinks]{hyperref}
\usepackage{tabularx}
\usepackage{booktabs}
\usepackage{tikz}
\usepackage{tikz-qtree}

\newcommand{\prxmlclass}[1]{\mathsf{#1}}
\newcommand{\fie}{\prxmlclass{fie}}
\newcommand{\cie}{\prxmlclass{cie}}
\newcommand{\mie}{\prxmlclass{mie}}
\newcommand{\muxind}{\prxmlclass{mux, ind}}
\newcommand{\mux}{\prxmlclass{mux}}
\newcommand{\ind}{\prxmlclass{ind}}
\newcommand{\inddrm}{\prxmlclass{ind, det}}
\newcommand{\muxdrm}{\prxmlclass{mux, det}}
\newcommand{\muxinddrm}{\prxmlclass{mux, ind, det}}
\newcommand{\drm}{\prxmlclass{det}}
\newcommand{\xpn}{\prxmlclass{exp}}

\newcommand{\tot}{<}
\newcommand{\emp}{\not<}

\newcommand{\pposs}{\textsc{Poss}\xspace}
\newcommand{\ppossc}{\#\pposs}
\newcommand{\pposse}{\textsc{EPoss}\xspace}
\newcommand{\ppossce}{\#\pposse}
\newcommand{\poss}[2]{\pposs_{#1}^{#2}}
\newcommand{\possc}[2]{\#\pposs_{#1}^{#2}}

\newcommand{\posse}[2]{\pposse_{#1}^{#2}}
\newcommand{\possce}[2]{\#\pposse_{#1}^{#2}}

\newcommand{\prxml}{\mathsf{PrXML}}

\newcommand{\fpsp}{\text{FP}^{\text{\#P}}}

\DeclareMathOperator{\supp}{supp}

\newcommand{\true}{\mathfrak{t}}
\newcommand{\false}{\mathfrak{f}}

\begin{document}

\author{Antoine Amarilli\vspace{-0.6em}}
\institute{Télécom ParisTech; Institut Mines-Télécom; CNRS LTCI}

\maketitle
\pageonenum

\begin{abstract}
  \vspace{-0.8em}
  We consider the \emph{possibility problem} of determining if a
  document is a possible world of a probabilistic document, in the setting of
  probabilistic XML. This basic question is a special case of query
  answering or tree automata evaluation, but it has specific practical
  uses, such as checking whether an user-provided probabilistic document outcome
  is possible or sufficiently likely.

  In this paper, we study the complexity of the possibility problem for
  probabilistic XML models of varying expressiveness. We show that the decision
  problem is often tractable in the absence of long-distance dependencies, but
  that its computation variant is intractable on unordered documents. We also
  introduce an \emph{explicit matches} variant to generalize practical
  situations where node labels are unambiguous; this ensures tractability of the
  possibility problem, even under long-distance dependencies, provided event
  conjunctions are disallowed. Our results entirely classify the tractability
  boundary over all considered problem variants.
  \vspace{-0.6em}
\end{abstract}

\section{Introduction}

Probabilistic representations are a way to represent incomplete knowledge
through a concise description of a large set of possible worlds annotated
with their probability. Such models can then be used, e.g., to run a query
efficiently over all possible worlds and determine the overall probability that
the query holds.
Probabilistic representations have been successfully used both for the relational
model~\cite{suciu2011probabilistic} and for XML
documents~\cite{kimelfeld2013probabilistic}.

Many problems, such as query answering~\cite{kimelfeld2009query},
have been studied over such representations; however,
to our knowledge, the
\emph{possibility problem} ($\pposs$) has not been specifically studied: given a
probabilistic document~$D$ and a deterministic document~$W$, decide if $W$
is a possible world of~$D$, and optionally compute its probability according to~$D$. 
This can be asked both of relational and XML probabilistic
representations, but we focus on XML documents
because they pose many challenges: they are hierarchical so some
probabilistic choices appear dependent\footnote{In fact, we will see that our
hardness results always hold even for shallow documents.}; documents may be
ordered; bag semantics must be used to count multiple sibling nodes with the
same label.
In addition, in the XML setting, the $\pposs$ problem is a
natural question that arises in practical scenarios.

As a first example, when using probabilistic XML to represent a set
$D$ of possible versions~\cite{ba2013uncertain} of an XML document, one may want
to determine if a version~$W$, obtained from a user or from an external source,
is one of the known possible versions represented as a probabilistic XML
document~$D$.
For instance, assume that a probabilistic XML
version control system asks a user to resolve a conflict~\cite{ba2013merging}, whose
uncertain set of possible outcomes is represented by $D$. When the user provides
a candidate merge~$W$, the system must then check if 
the document $W$ is indeed a possible way to solve the
conflict.
This may be hard to
determine, because $D$ may, in general, have many ways to generate~$W$, through a
possibly intractable number of different valuations of its uncertainty events.

As a second practical example, assume that a user is studying an uncertain
document~$D$ that provides a representation of possible versions of an XML tree,
using probabilistic XML to represent possible conflicting choices and their
probability. The user notices that choosing a certain combination of decisions
yields a certain deterministic document~$W$, and asks whether the same document
could have been obtained by making different choices. Indeed, maybe $W$ is
considered improbable under $D$ following this particular valuation, but is
likely overall because the same document can be obtained through many different
ways. What is the probability, over all valuations, of the user's chosen
outcome $W$ according to $D$?

On the face of it, $\pposs$ seems related to query evaluation: we wish to
evaluate on $D$ a query $q_W$ which is, informally, ``is the input
document exactly $W$''? However, there are three reasons why query evaluation
cannot give good complexity bounds for $\pposs$. First, because $q_W$ depends
on the possibly large $W$, we are not performing query answering for a
\emph{fixed} query, so we can only use the unfavorable \emph{combined
complexity} bounds where both the input document $D$ and the query $q_W$ are
part of the input. Second, because we want to obtain \emph{exactly} $W$, the
match of $q_W$ should never map two query variables to the same node
of $D$, so the query language must allow inequalities on node identifiers.
Third, once again because we require an exact match, we need to assert the
\emph{absence} of the nodes which are not in~$W$, so we need negation in the 
language.
To our knowledge, then, the only upper bound for $\pposs$ from query answering
is the combined complexity bound for the (expressive) \emph{monadic second-order
logic over trees} whose evaluation on deterministic (not even probabilistic) XML
trees is already PSPACE-hard~\cite{libkin2004elements}.

A second related approach is that of tree automata on probabilistic XML
documents. Indeed, we can encode the possible world~$W$ to a
deterministic tree automaton~$A_W$ and compute the probability that $A_W$
accepts the
probabilistic document $D$. The decision and computation variants of $\pposs$
under local uncertainty models 
are thus special cases of the ``relevancy'' and ``p-acceptance'' problems
of~\cite{cohen2009running}. However, their work only considers
\emph{ordered} trees, and an unordered $W$ cannot easily be translated to their
deterministic tree
automata, because of possible label ambiguity: we
cannot impose an arbitrary order on~$D$ and~$W$, as this also chooses how
nodes must be disambiguated. In fact, we will show that $\pposs$ is hard in
some settings that are tractable for ordered documents.

This paper specifically focuses on the $\pposs$ problem to study the precise complexity
of its different formulations.
Our probabilistic XML representation is the $\prxml$ model
of~\cite{kimelfeld2013probabilistic}, noting that some results are known for the
$\pposs$ problem (called the ``membership problem'') in the incomparable and
substantially different ``open-world'' incomplete XML model
of~\cite{barcelo2010xml} (whose documents have an infinite set of possible
worlds, instead of a possibly exponential but finite set as in $\prxml$).

We start by defining the required preliminaries in Section~\ref{sec:prelim} and
the different variants of $\pposs$ in Section~\ref{sec:problem},
establishing its overall NP-completeness and reviewing the results
of~\cite{cohen2009running}. We then study local uncertainty models in
Section~\ref{sec:local} and show that the absence of order impacts
tractability, with a different picture for the decision and computation variants
of $\pposs$. Last, in Section~\ref{sec:condition}, we show that $\pposs$ can be
made tractable under long-distance event correlations, by disallowing event
conjunctions and imposing an ``explicit matches'' condition which generalizes,
e.g., unique node labels.
We then conclude in Section~\ref{sec:conclusion}. \omitteddisc

\section{Preliminaries}
\label{sec:prelim}

We start by formally defining XML documents and probability distributions over
them:

\begin{definition}
  An \emph{unordered XML document} is an unordered tree whose nodes carry a
  label from a set $\Lambda$ of labels. \emph{Ordered} XML documents are defined
  in the same way but with ordered trees, that is, there is a total order over
  the children of every node.

  A \emph{probability distribution} is a function $\mathcal{P}$ mapping every
  XML document $x$ from a finite
  set $\supp(\mathcal{P})$ to a rational number
  $\mathcal{P}(x)$, its \emph{probability} according to $\mathcal{P}$,
  with the condition that $\sum_{D \in \supp(\mathcal{P})} \mathcal{P}(D) = 1$.
  For any $x \notin \supp(\mathcal{P})$ we write $\mathcal{P}(x) = 0$.
\end{definition}

As it is unwieldy to manipulate explicit probability distributions over large
sets of documents, we use the language of probabilistic
XML~\cite{kimelfeld2013probabilistic} to write extended
XML documents (with so-called \emph{probabilistic nodes}) and give them a
semantics which is a
(possibly exponentially larger) probability distribution over XML documents.
Intuitively, probabilistic XML documents are XML documents with specific
\emph{probabilistic} nodes describing possible choices in the document; their
semantics is the set of XML documents that can be obtained under those choices.

\begin{definition}
  A $\prxml$ \emph{probabilistic XML document} $D$ is an XML document over $\Lambda
  \sqcup \{\drm, \ind, \mux, \cie, \fie\}$.
  The nodes of~$D$ with
  labels from $\Lambda$ are called \emph{regular} nodes, by opposition to
  \emph{probabilistic} nodes.
  The probabilistic labels 
  respectively stand for: determininistic, independent, mutually exclusive,
  conjunction of independent events, formula of independent events.
  For any subset $\mathcal{L} \subseteq \{\drm, \ind, \mux, \cie, \fie\}$, we
  call $\prxml^{\mathcal{L}}$ the language of probabilistic XML documents
  containing only nodes with labels in $\Lambda \sqcup \mathcal{L}$.

  We require that the root of a $\prxml$ document~$D$ is a regular node,
  that every edge from a $\mux$ or
  $\ind$ node to a child node is labeled with some rational
  number\footnote{The non-standard constraint $x < 1$ means that
    $\ind$ does not subsume $\drm$ (see Thm.~\ref{thm:possmuxind}
    and~\ref{thm:possempmuxinddrm} for examples where this distinction matters).}
    $0 < x < 1$
  (the sum of the labels of the children of every $\mux$ node being $\leq 1$),
  and that every edge from a $\cie$ (resp.\ $\fie$) node to a child node is labeled
  with a conjunction (resp.\ a Boolean formula) of \emph{events} from a set $E$
  of events (and their negations), with~$D$ providing a mapping $\pi : E \rightarrow [0,
  1]$ attributing a rational probability to every event.
\end{definition}

The semantics of a $\prxml$ document $D$ is the probability distribution over
XML documents defined by the following sampling process
(see~\cite{kimelfeld2013probabilistic} for more details):

\begin{definition}
  A deterministic XML document $W$ is obtained from a $\prxml$ document $D$ as
  follows. First, choose a valuation $\nu : E \rightarrow \{\true, \false\}$ of
  the events from $E$, with probability $\prod_{e~\mathrm{s.t.} \nu(e) =
  \true} \pi(e) \times \prod_{e~\mathrm{s.t.} \nu(e) = \false} (1 -
  \pi(e))$. Evaluate $\cie$ and $\fie$ nodes by keeping only the child
  edges whose Boolean formula is true under $\nu$. Evaluate $\ind$ nodes by
  choosing to keep or delete every child edge according to the probability
  indicated on its edge label. Evaluate $\mux$ nodes by removing all of their
  children edges, except one chosen according to its probability (possibly
  keep none if the probabilities sum up to less than~$1$). Finally, evaluate
  $\drm$ nodes by replacing them by the collection of their
  children.

  All probabilistic choices are performed independently,
  so the overall probability of an outcome is the product of the
  probabilities at each step. Whenever an edge is removed, all of the
  descendant nodes and edges are removed. The probability of a document
  $W$ according to $D$, written $D(W)$, is the total probability of all
  outcomes\footnote{Note that in general there may be multiple
  outcomes that lead to the same document $W$.} leading to~$W$.
\end{definition}

\begin{figure}[t]
\begin{flushright}
  \begin{tabular}{ll}
    \toprule
    {\bf Event~~} & {\bf Probability} \\
    \midrule
    $e$ & $0.9$ \\
    \bottomrule
  \end{tabular}
\end{flushright}
\vspace{-5.2em}
  \centering
\begin{tikzpicture}[
  level 2/.style={sibling distance=1em},
  level 3/.style={sibling distance=-2em},
  level distance = 3em
]
\Tree
[.conferences
  [.{$\ind$}
    \edge node[auto=right] {0.8};
    [.conference
      [.name
        [.{BDA} ]
      ]
      [.{$\cie$}
        \edge node[auto=left] {$e$};
        [.location
          [.city
            [.Grenoble-Autrans ]
          ]
          [.country
            [.FR ]
          ]
        ]
      ]
    ]
    \edge node[auto=left] {0.7};
    [.conference
      [.name
        [.{AMW} ]
      ]
      [.{$\cie$}
        \edge node[auto=left] {$e$};
        [.location
          [.{$\mux$}
            \edge node[auto=right] {0.9};
            [.{$\drm$}
              [.city
                [.{Cartagena de Indias} ]
              ]
              [.country
                [.CO ]
              ]
            ]
            \edge node[auto=left] {0.1};
            [.{$\drm$}
              [.city
                [.{Cartagena} ]
              ]
              [.country
                [.ES ]
              ]
            ]
          ]
        ]
      ]
    ]
  ]
]
\end{tikzpicture}
\caption{Example $\prxml^{\muxinddrm,\cie}$ document; the provided table is the
mapping~$\pi$ that attributes probabilities to probabilistic events}
\label{fig:example}
\end{figure}
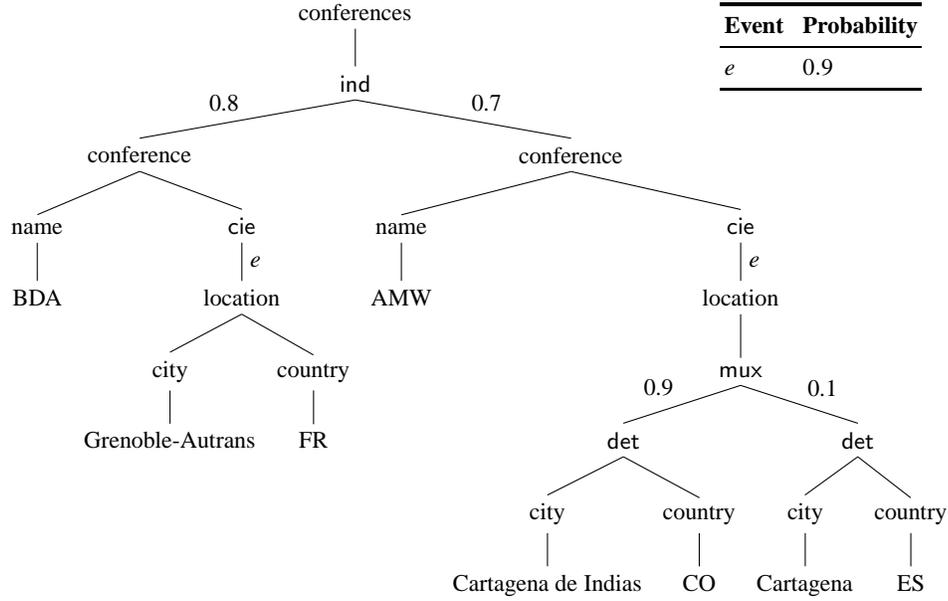

We say that $\mux$, $\ind$ and $\drm$ are \emph{local} in the sense that they
describe a probabilistic choice that \emph{takes place} at this point of the
document, independently from other choices (except for the fact that discarding
a subtree makes irrelevant all local probabilistic choices in that subtree. By
contrast, we say that $\cie$ and $\fie$ are \emph{long-distance} in the sense
that a valuation is chosen \emph{globally} for the probabilistic events and the $\cie$
and $\fie$ nodes are then evaluated according to that choice: this may induce
\emph{correlations} between arbitrary portions of the document, because the
same event can be reused multiple times at different positions in the document.

\begin{example}
  Consider the example probabilistic XML document~$D$ in
  Figure~\ref{fig:example}. Its possible worlds are obtained as follows. First,
  draw a valuation for the (only) event $e$, which may be~$\false$ (with
  probability~$0.1$) or~$\true$ (with probability~$0.9$). Then, decide whether
  to keep or discard the first ``conference'' subtree, with probability~$0.8$,
  and decide whether to keep or discard the second such subtree, with
  probability~$0.7$. Remove the $\cie$ nodes and keep or discard their children
  depending on whether the chosen valuation for~$e$ is~$\true$ or~$\false$
  respectively. Decide whether to keep the first or
  second child of the $\mux$ node, and replace the corresponding $\drm$ node by
  its children. All probabilistic choices are made independently.

  Observe how the choice on $\mux$ is irrelevant if the corresponding subtree
  was discarded by the parent $\ind$ node or by the $\cie$ node, and notice the
  use of $\drm$ nodes to switch between sets of nodes using a $\mux$ node. Note
  that the use of $\cie$ nodes introduces a \emph{correlation} in the sense that
  the first ``location'' node is present if and only if the second is also
  present.
\end{example}

Of course, the expressiveness and compactness of $\prxml$ frameworks depend on
which probabilistic nodes are allowed: we say that $\prxml^{\mathcal{C}}$ is \emph{more
general} than $\prxml^{\mathcal{D}}$ if there is a polynomial time algorithm to
rewrite any $\prxml^{\mathcal{D}}$ document to a $\prxml^{\mathcal{C}}$ document
representing the same probability distribution. 
Fig.~\ref{fig:hier}
(adapted from~\cite{kharlamov2010updating}) represents this hierarchy
on the $\prxml$ classes that we consider.

\begin{figure}[t]
\noindent
\begin{minipage}{0.515\textwidth}
  \setlength{\tabcolsep}{2pt}
\begin{tabular}{@{\hskip -0.2pt}r@{\hskip 2.5pt}c@{\hskip 2.5pt}l@{\hskip 4pt}l@{\hskip 4pt}r}
\toprule
\multicolumn{3}{c}{\bf \hspace{-3em}Problem} &
\multicolumn{2}{c}{\bf Complexity} \\
\midrule
$\pposs$ & $\top$ & $\fie$ & NP & (Prop.~\ref{prop:uposs-upossc}) \\
$\ppossc$ & $\top$ & $\fie$ & $\fpsp$ & (Prop.~\ref{prop:uposs-upossc}) \\
$\ppossc$ & $\tot$ & $\muxinddrm$ & PTIME & (Thm.~\ref{thm:possctotmuxinddrm}) \\
$\ppossc$ & $\emp$ & $\ind$ or $\mux$ & $\#$P-hard & (Thm.~\ref{thm:posscemp}) \\
$\pposs$ & $\top$ & $\ind$ or $\mux$ & PTIME & (Thm.~\ref{thm:possmuxind})\\
$\pposs$ & $\emp$ & 2 of $\muxinddrm$ & NP-hard & (Thm.~\ref{thm:possempmuxinddrm})\\
$\ppossce$ & $\top$ & $\muxinddrm$ & PTIME & (Thm.~\ref{thm:posscemuxinddrm}) \\
$\pposse$ & $\bot$ & $\cie$ & NP-hard & (Thm.~\ref{thm:possecie}) \\
$\pposs$ & $\bot$ & $\mie$ & NP-hard & (Thm.~\ref{thm:possmie}) \\
$\ppossce$ & $\top$ & $\mie$ & PTIME & (Thm.~\ref{thm:posscemie}) \\
\bottomrule
\end{tabular}

  \vspace{-1em}
  \captionof{table}{Summary of results}
  \label{tab:results}
\end{minipage}
\hfill\begin{minipage}{0.49\textwidth}
\noindent
\hspace{-.05\textwidth}
\begin{minipage}{0.43\textwidth}
  \centering
  \begin{tikzpicture}[>=latex,line join=bevel,scale=0.5]
\node (TOP) at (54bp,106bp) [draw,draw=none] {$\top$};
  \node (BOT) at (54bp,18bp) [draw,draw=none] {$\bot$};
  \node (EMP) at (82bp,62bp) [draw,draw=none] {$\emp$};
  \node (TOT) at (27bp,62bp) [draw,draw=none] {$\tot$};
  \draw [<-] (TOT) ..controls (42.735bp,36.523bp) and (42.838bp,36.363bp)  .. (BOT);
  \draw [<-] (TOP) ..controls (70.318bp,80.523bp) and (70.424bp,80.363bp)  .. (EMP);
  \draw [<-] (TOP) ..controls (38.265bp,80.523bp) and (38.162bp,80.363bp)  .. (TOT);
  \draw [<-] (EMP) ..controls (65.682bp,36.523bp) and (65.576bp,36.363bp)  .. (BOT);
\end{tikzpicture}
  \smallskip
  \centering
  \hspace{.07\textwidth}\begin{minipage}{.93\textwidth}
    \centering
  \begin{tikzpicture}[>=latex,line join=bevel,scale=0.5]
\node (POSSCE) at (86bp,62bp) [draw,draw=none] {$\ppossce$};
  \node (POSSE) at (56bp,18bp) [draw,draw=none] {$\pposse$};
  \node (POSSC) at (56bp,106bp) [draw,draw=none] {$\ppossc$};
  \node (POSS) at (27bp,62bp) [draw,draw=none] {$\pposs$};
  \draw [<-] (POSSCE) ..controls (68.516bp,36.523bp) and (68.402bp,36.363bp)  .. (POSSE);
  \draw [<-] (POSS) ..controls (43.901bp,36.523bp) and (44.011bp,36.363bp)  .. (POSSE);
  \draw [<-] (POSSC) ..controls (39.099bp,80.523bp) and (38.989bp,80.363bp)  .. (POSS);
  \draw [<-] (POSSC) ..controls (73.484bp,80.523bp) and (73.598bp,80.363bp)  .. (POSSCE);
\end{tikzpicture}
  \end{minipage}
\end{minipage}
\hspace{-.13\textwidth}
\begin{minipage}{0.47\textwidth}
  \centering
  \medskip
  \begin{tikzpicture}[>=latex,line join=bevel,scale=0.47]
\node (MUXINDDRM) at (117bp,150bp) [draw,draw=none] {$\muxinddrm$};
  \node (MUXDRM) at (190bp,106bp) [draw,draw=none] {  $\muxdrm$  };
  \node (CIE) at (171bp,194bp) [draw,draw=none] {$\cie$};
  \node (MIE) at (225bp,150bp) [draw,draw=none] {\phantom{p}$\mie$};
  \node (FIE) at (171bp,238bp) [draw,draw=none] {$\fie$};
  \node (INDDRM) at (36bp,106bp) [draw,draw=none] {  $\inddrm$  };
  \node (MUX) at (178bp,62bp) [draw,draw=none] {$\mux$};
  \node (EMP) at (144bp,18bp) [draw,draw=none] {$\emptyset$};
  \node (MUXIND) at (111bp,106bp) [draw,draw=none] {  $\muxind$  };
  \node (IND) at (111bp,62bp) [draw,draw=none] {$\ind$};
  \draw [<-] (MUXINDDRM) ..controls (157.04bp,125.96bp) and (158.58bp,125.08bp)  .. (MUXDRM);
  \draw [<-] (MUXINDDRM) ..controls (72.999bp,126.18bp) and (70.945bp,125.12bp)  .. (INDDRM);
  \draw [<-] (MUXINDDRM) ..controls (113.5bp,124.52bp) and (113.48bp,124.36bp)  .. (MUXIND);
  \draw [<-] (MIE) ..controls (233.57bp,110.57bp) and (233.08bp,97.929bp)  .. (227bp,88bp) .. controls (222.06bp,79.937bp) and (213.62bp,74.372bp)  .. (MUX);
  \draw [<-] (IND) ..controls (130.23bp,36.523bp) and (130.36bp,36.363bp)  .. (EMP);
  \draw [<-] (CIE) ..controls (140.35bp,169.16bp) and (139.73bp,168.68bp)  .. (MUXINDDRM);
  \draw [<-] (MUXDRM) ..controls (183.01bp,80.523bp) and (182.96bp,80.363bp)  .. (MUX);
  \draw [<-] (INDDRM) ..controls (78.181bp,81.379bp) and (80.957bp,79.824bp)  .. (IND);
  \draw [<-] (CIE) ..controls (201.65bp,169.16bp) and (202.27bp,168.68bp)  .. (MIE);
  \draw [<-] (FIE) ..controls (171bp,212.52bp) and (171bp,212.36bp)  .. (CIE);
  \draw [<-] (MUXIND) ..controls (148.35bp,81.585bp) and (149.68bp,80.754bp)  .. (MUX);
  \draw [<-] (MUXIND) ..controls (111bp,80.523bp) and (111bp,80.363bp)  .. (IND);
  \draw [<-] (MUX) ..controls (158.19bp,36.523bp) and (158.06bp,36.363bp)  .. (EMP);
\end{tikzpicture}
\end{minipage}
\vspace{0.15em}
\captionof{figure}{Variants and PTIME reductions}
  \label{fig:hier}
\end{minipage}
\vspace{-1em}
\end{figure}

\section{Problem and general bounds}
\label{sec:problem}

We now define the $\pposs$ problem formally, in its decision and computation
variants.

\begin{definition}
  Given a class $\prxml^{\mathcal{C}}$, the \emph{possibility problem} for
  \emph{unordered} documents $\poss{\emp}{\mathcal{C}}$ is to determine, given
  as input an unordered $\prxml^{\mathcal{C}}$ document $D$ and an unordered XML
  document $W$, whether $W$ is a possible world of $D$, namely, $D(W) > 0$.

  The \emph{possibility problem} for \emph{ordered} documents
  $\poss{\tot}{\mathcal{C}}$ is the same problem except that both $D$ and $W$ are
  ordered.
  For $o \in \{\emp, \tot\}$,
  the $\possc{o}{\mathcal{C}}$ problem
  is to compute the probability~$D(W)$ of~$W$ according to~$D$. Observe that
  $\possc{o}{\mathcal{C}}$ is a \emph{computation problem} rather than a
  decision problem, namely, it computes an output value based on the provided
  input (here, a probability value) instead of merely deciding whether to accept
  or reject.
\end{definition}

For brevity, we write $\poss{\bot}{\mathcal{C}}$ and $\poss{\top}{\mathcal{C}}$
when describing lower or upper complexity bounds that apply to both
$\poss{\tot}{\mathcal{C}}$ and $\poss{\emp}{\mathcal{C}}$.

We start by giving straightforward bounds on the most general problem variants:

\begin{proposition}
  \label{prop:uposs-upossc}
  $\poss{\top}{\fie}$ is in NP and $\possc{\top}{\fie}$ is in $\fpsp$.
\end{proposition}

\begin{proof}
  We first show the NP-membership of $\poss{\top}{\fie}$.

  Let us first consider $\poss{\tot}{\fie}$.
  Consider the input $(D, W)$. Guess a valuation of the probabilistic events of
  $D$. The size of the guess is linear in $|D|$. Now, check that the guess is
  suitable, namely, that the deterministic document $D'$ obtained from $D$ under
  this valuation is exactly $W$: as both $D'$ and $W$ are totally ordered trees,
  this can be checked straightforwardly in linear time from a simultaneous
  traversal of $D'$ and $W$. Hence, $\poss{\tot}{\fie}$ is in NP.

  Let us now consider  $\poss{\emp}{\fie}$. The proof idea is the same, except
  that checking that $D'$ and $W$ are equal is not as obvious, because those
  trees are not ordered; however, this check can be performed in PTIME by a
  dynamic bottom-up algorithm similar to that of the proof of
  Thm.~\ref{thm:possmuxind}, so that the result still holds.

  \medskip

  We now show the $\fpsp$-membership of $\possc{\top}{\fie}$.
  
  We first preprocess all the event probabilities in the probabilistic document
  $D$ so that all numbers are represented with the same denominator. This can be
  done in polynomial time by a least common multiple computation and product
  operations. We then read off the common denominator, $d$. We can compute $d^k$
  in polynomial time, where $k$ is the number of events.

  We then use Lemma~5.2 of~\cite{abiteboul2011capturing} to argue that it is
  possible, in \#P, to compute the unnormalized probability of $W$, that is,
  the probability of $W$ in $D$ without dividing by $d^k$. To do so, the
  generating PTIME Turing machine $T$ enumerates all possible valuations, and
  the function $g$ returns $0$ if the outcome does not yield the desired
  document $W$ (which can be decided in PTIME by the above proof for the
  decision problem), and otherwise returns the unnormalized probability of the
  outcome, that is, the product of the numerators of the involved probabilities.
  (The denominator, which would be $d^k$, is ignored for now.) Hence, by
  application of this lemma, $\possc{\top}{\fie}$ is in $\fpsp$, because all
  that remains is to divide the result of this \#P computation by $d^k$ to
  obtain the final probability.
\end{proof}

\begin{proposition}
  \label{prop:posscie}
  $\poss{\bot}{\cie}$ is NP-complete, even when $D$ has height $3$.
\end{proposition}

\begin{proof}
  From Prop.~\ref{prop:uposs-upossc} it suffices to show hardness. We show a reduction
  from the NP-hard Boolean satisfiability problem to justify that $\poss{\bot}{\cie}$ 
  is NP-hard.

  Consider a formula $F$ formed of a conjunction of disjunctive clauses
  $(C_i)_{1 \leq i \leq n}$, with clause~$C_i$ containing the literals
  $(l^i_j)_{1 \leq j \leq n_i}$, each literal being a positive or negative
  occurrence of some variable from a finite set of variables $V = \{x_1, \ldots,
  x_m\}$.

  Consider a set of $m$ Boolean events $E$ with a mapping $\phi$ associating
  $x_i \in V$ to $e_i \in E$ (and $\neg x_i$ to $\neg e_i$) for all $1 \leq i \leq m$.
  Consider $W$ the document with only one root labeled~$\top$, and the
  $\prxml^{\cie}$ document $D$, with events $E$ (and probability $1/2$ for each
  outcome), with one root labeled $\top$
  and one $\cie$ child with $n$ children labeled $\bot$, the edge of the $i$-th
  child being labeled with
  $C_i' = \neg \phi(l^i_1) \wedge \cdots \wedge \neg \phi(l^i_{n_i})$. Given the
  shape of $W$, clearly the algorithm's choice to consider $D$ and $W$ as either
  ordered or unordered trees is irrelevant, so that this works as a reduction
  either to $\poss{\emp}{\cie}$  or to $\poss{\tot}{\cie}$.

  Now, $W$ is a possible world of $D$ if and only if there is a valuation of
  the events of $E$ such that $\bigwedge_i \neg C_i'$ holds, or, equivalently by
  De Morgan's law, such that $\bigwedge_i \bigvee_j \phi(l^i_j)$ holds, hence
  $(D, W)$ is a positive instance of $\poss{\bot}{\cie}$ if and only if $F$ is
  satisfiable. Hence, $\poss{\emp}{\cie}$ is NP-hard.
\end{proof}

Local models on ordered documents are known to be tractable using tree
automata:

\begin{theorem}[\cite{cohen2009running}]
  \label{thm:possctotmuxinddrm}
  ~$\possc{\tot}{\muxinddrm}$ can be solved in polynomial time.
\end{theorem}

\begin{proof}
  We prove the theorem using the results of~\cite{cohen2009running}. An
  alternative, stand-alone proof is given in
  Appendix~\ref{apx:possctotmuxinddrmproof}.

  The input $\prxml^{\muxinddrm}$ document $D$ can be rewritten to an equivalent
  $\prxml^{\xpn}$ document in polynomial time~\cite{abiteboul2009expressiveness}, which is
  a pTT document as defined by~\cite{cohen2009running} (note that we make no use
  of the possibility of having uncertainty about order).
  
  Furthermore, we can encode the deterministic document $W$ to a deterministic
  tree automaton $A_W$ with deterministic finite-state automata describing the
  regular languages of the transition function. Informally, the various states
  of the automaton will correspond to the various subtrees of $W$, except that
  subtrees occurring multiple times need to be identified. Formally, we define
  an equivalence relation~$\sim$ on the nodes of~$W$ with $v \sim w$ if
  the subtrees rooted at $v$ and $w$ are isomorphic (i.e., they are the same
  tree, taking order into account. Let $C_W$ be the set of classes of this relation, and $\phi$ be a
  mapping from the nodes of~$W$ to their class in $C_W$. We can use a dynamic
  bottom-up algorithm similar to that of the proof of Thm.~\ref{thm:possmuxind}
  to compute the $\sim$ relation in polynomial time, as well as $C_W$ and the mapping
  $\phi$. Now, the alphabet of the automaton~$A_W$ is the set of node labels~$\Lambda$, its
  set of states is~$C_W$, its accepting state is~$\phi(r)$ where $r$ is the root
  of~$W$, and its transition function maps $(c, l) \in C_W \times \Lambda$ to
  the empty language (if $l$ is not the label of the nodes in $c$, noting that their
  labels must coincide) or (otherwise) to the language consisting of the single
  word $c_1 \cdots c_n$ where $n$ is the number of children of all nodes~$v$
  of~$W$ in the class $c$ and, for all~$i$, $c_i$ is the class of the $i$-th
  child (note that $n$ and the $c_i$ do not depend on the choice of
  representative~$v$). Computing $A_W$, with the languages of the transition
  function being represented by a deterministic finite-state automaton, can be
  done in polynomial time, and clearly by induction $A_W$ accepts a tree $T$ if
  and only if $T$ is isomorphic to~$W$.
  
  The problem $\possc{\tot}{\muxinddrm}$ then amounts to computing the total
probability of the possible worlds of~$D$ that are accepted by $A_W$, which can
be computed in polynomial time by Theorem~2 of~\cite{cohen2009running}.
\end{proof}

\section{Local models}
\label{sec:local}

We now complete the picture for the local model $\prxml^{\muxinddrm}$ on
unordered documents. The results of~\cite{cohen2009running} cannot be applied to
this setting, as the ambiguity of node labels imply that we cannot impose an
arbitrary order on document nodes; indeed, a reduction from perfect matching
counting on bipartite graphs shows that the computation variant is hard even on
the most inexpressive classes:

\begin{theorem}
  \label{thm:posscemp}
  $\possc{\emp}{\ind}$ and $\possc{\emp}{\mux}$ are \#P-hard, even when $D$
  has height $4$.
\end{theorem}

\begin{proof}
  We first focus on the case of $\prxml^{\ind}$. 
  We show a reduction from the problem of counting the number of perfect
  matchings
  of a bipartite graph\footnote{Recall that a \emph{perfect matching} in a bipartite graph is a subset of its edges
    such that each vertex of the graph (in either part) is adjacent to
  exactly one edge of the subset.}, which is 
  \#P-hard~\cite{valiant1979complexity}. Let $G = (V, W, E)$ be a bipartite
  graph. We assume without loss of generality that $|V| = |W|$ (as $G$ certainly
  cannot have perfect matchings otherwise), and let $n = |V| = |W|$.

  Now, consider $W$ with root labeled $\top$, $n$ children labeled $\bot$, each of
  them with one child with labels respectively $l_1, \ldots, l_n$.
  Consider the uncertain document $D$ with root labeled $\top$, $n$ children labeled
  $\bot$, the $i$-th of them (for all $i$) having, for every $j$ such that there
  is an edge in~$E$ from node $i$ of~$V$ to node $j$ of $W$, an $\ind$ child with one
  child labeled $l_j$ with edge label $1/2$.

  We claim that $D(W)$ is exactly the number of perfect matchings of the
  bipartite graph~$G$, divided by $F = 2^{|E|}$.

  To see why this is true, notice that each edge of $E$ corresponds to an $\ind$
  node of~$D$. Hence, for any subset $M \subseteq E$, let us consider the
  valuation $\nu_M$ where the $\ind$ nodes for edges in $M$ keep their child
  node, and the $\ind$ nodes for edges not in $M$ discard their child node. This
  mapping between subsets of $E$ and valuations is clearly one-to-one, and all
  those valuations have probability $1/F$ (because each of the $|E|$ events has
  probability $1/2$ and all of them are independent).

  It only remains to see that the valuation $\nu_M$ yields $W$ if and only if
  $M$ is a perfect matching, but this is easy to see: if $M$ is a perfect
  matching, each node labeled $\bot$ keeps exactly one child, and one node
  labeled $l_i$ is kept for each node, so that $\nu_M$ yields $W$; conversely,
  if $M$ is not a perfect matching, either there is a node labeled $\bot$ with
  zero or~$>1$ children, or there is some $l_i$ kept zero or~$>1$ times, so that
  $\nu_M$ does not yield $W$.
  Hence, this completes the reduction, and shows that $\possc{\emp}{\ind}$ is
  \#P-hard.

  For the case of $\prxml^{\mux}$, observe that the previous proof can be
  immediately adapted by replacing $\ind$ nodes with $\mux$ nodes, as every
  $\ind$ node has exactly one child.
\end{proof}

By contrast, the decision variant is tractable for $\prxml^{\ind}$ and
$\prxml^{\mux}$, using a dynamic algorithm. However, allowing both $\ind$ and
$\mux$, or allowing $\drm$ nodes, leads to intractability (by reductions from
set cover and Boolean satisfiability).

\begin{theorem}
  \label{thm:possmuxind}
  $\poss{\top}{\ind}$ and $\poss{\top}{\mux}$ can be decided in PTIME.
\end{theorem}

\begin{proof}
  For ordered documents, the result follows from
  Theorem~\ref{thm:possctotmuxinddrm}, so we only prove the claim that
  $\poss{\emp}{\ind}$ and $\poss{\emp}{\mux}$ can be decided in PTIME.
  
  We show a stronger result, namely: the $\poss{\emp}{\muxind}$ problem can be
  decided in PTIME under the assumption that no $\ind$ node is a child of a
  $\mux$ node. Note that under this assumption, subtrees of $D$ rooted at nodes
  that are not $\ind$ nodes only have possible worlds that are (possibly empty)
  subtrees (by contrast, $\ind$ nodes may have possible worlds that are
  forests). We say that a node of $D$ is \emph{non-$\ind$} if it is a regular
  node or a $\mux$ node.
  
  We will present a dynamic algorithm to decide $\poss{\emp}{\muxind}$ in PTIME
  under this assumption. We first compute bottom-up, for every non-$\ind$
  node~$n$ of $D$, a Boolean value $e(n)$ indicating whether the subtree of $D$
  rooted at $n$ has an empty possible world. If $n$ is a regular node, we define
  $e(n) = \false$. If $n$ is a $\mux$ node, we define $e(n) = \true$ if the
  probabilities of $n$ sum up to $<1$, or if one child $n'$ of $n$ is such that
  $e(n') = \true$. It is clear that this computation can be performed in
  polynomial time.

  The algorithm will now compute bottom-up, for every pair $(n, n')$ of a
  non-$\ind$ node~$n$ in~$D$ and a node~$n'$ in~$W$, a Boolean value $c(n, n')$
  indicating whether or not the subtree of $W$ rooted at $n'$ is a possible
  world of the subtree of $D$ rooted at $n$.

  If $n$ is a regular leaf, we define $c(n, n') = \true$ if $n'$ is a leaf with
  the same label as $n$, and $c(n, n') = \false$ otherwise. Note that we can
  assume without loss of generality that all of $D$'s leaves are regular nodes,
  as leaves that are probabilistic nodes can simply be removed.

  If $n$ is a $\mux$ node, we define $c(n, n') = \true$ if one of the
  children~$x$ of $n$ is such that $c(x, n')$ is $\true$, otherwise $c(n, n') =
  \false$. Observe that this is correct because the children of $n$ are either
  $\mux$ nodes or regular nodes (they cannot be $\ind$ nodes), so the possible
  worlds of $n$ are exactly the possible worlds of its children (possibly in
  addition to the empty subtree), and those possible worlds must be subtrees and not forests.

  If $n$ is an internal regular node of $D$, to define $c(n, n')$, we first
  check if $n$ and $n'$ have the same label. If they do not, we define $c(n, n')
  = \false$; otherwise we continue.

  Consider $D$ the set of the topmost non-$\ind$ descendants of $n$. We say that
  a node~$x$ of~$D$ is \emph{optional} if there is an $\ind$ node on the path
  from~$n$ to~$x$, or if $e(x) = \true$. In other words, a node $x$ is optional
  if there is a valuation (of $\ind$ nodes) that discards it, or if there is a
  valuation of the subtree rooted at $x$ which achieves an empty possible world
  for this subtree. This implies that, because the probabilistic choices are local
  and independent, we have a way to keep or delete every optional node of $D$
  \emph{independently of each other}. Call $D'$ the set of the children of $n'$
  in~$W$.

  Now if $|D| < |D'|$ we define $c(n, n') = \false$ (because in no possible
  world can $n$ have sufficiently many children to match $n'$ -- remember that
  the possible worlds of the subtrees of $D$ rooted at non-$\ind$ nodes must be
  (possibly empty) subtrees but cannot be forests). Otherwise, add $|D|-|D'|$
  \emph{dummy nodes} to $D'$ so that $|D'| = |D|$. Build a bipartite
  graph~$G_{n,n'} = (D, D', E)$ with edges $E$ defined as follows:
  an edge between $x \in D$ and non-dummy $x' \in D'$ if and only if $c(x, x')$ is $1$, and
  an edge between $x$ and dummy $x'$ if and only if $x$ was optional. (Intuitively: dummy
  nodes of~$D'$ represent the choice of deleting a node of $D$.)

  We now claim that we should define $c(n, n') = \true$ if and only if $G_{n, n'}$ has a perfect
  matching. To see why, observe first that $c(n, n')$ should be $\true$ if and only if the subtree
  of~$W$ rooted at~$n'$ is a possible world of the subtree of~$D$ rooted at~$n$,
  which, because $n$ is a regular node and the labels of $n$ and $n'$, amounts
  to saying that $D'$ is a possible world of~$D$. Observe now that for any
  subset $S$ of $E$ such that each vertex of $D$ has exactly one incident edge,
  $S$ describes a possible world of $D$: each node of $D$ can achieve the node
  of $D'$ to which it is thus matched (or, for dummy nodes, the empty subtree),
  because choices on the nodes of~$D$ (and their descendants, or at their parent
  edge in the case of deletions using an $\ind$ node) are independent between
  nodes of~$D$. Now, a perfect matching describes a possible world of $D$
  achieving exactly $D'$ (with no repetitions), and conversely if $D'$ is a
  possible world of $D$ it must be achieved by certain nodes of $D$ realizing
  the nodes of $D'$ (each node of $D'$ being realized exactly once), and the
  others being deleted (each one being matched to one of the dummy nodes) so
  $G_{n, n'}$ must have a perfect matching.

  Now, the existence of a perfect matching for the bipartite graph $G_{n, n'}$
  can be decided in PTIME (using, e.g., the Hopcroft-Karp algorithm), so we can
  decide how to define $c(n, n')$ in PTIME (with a fixed polynomial not
  dependent on~$n$ or~$n'$).

  Hence, we can decide in PTIME whether $W$ is a possible world of $D$,  by
  checking if $c(r, r')$ is $\true$, with $r$ and $r'$ the roots of $D$ and $W$
  (remember that $r$ is assumed not to be a probabilistic node). This dynamic
  algorithm considers a quadratic number of pairs, and performs a
  polynomial-time computation (with a fixed polynomial) for each of them, so its
  overall running time is polynomial.
\end{proof}

\begin{theorem}
  \label{thm:possempmuxinddrm}
  $\poss{\emp}{\inddrm}$, $\poss{\emp}{\muxdrm}$ and $\poss{\emp}{\muxind}$ are
  NP-complete, even when $D$~has height $4$.
\end{theorem}

\begin{proof}
  From Prop.~\ref{prop:uposs-upossc} it suffices to show hardness. 
  Let us first consider $\poss{\emp}{\inddrm}$.
  We show a reduction from the NP-hard~\cite{karp1972reducibility} exact cover problem.

  Consider an exact cover instance $S = \{S_1, \ldots, S_n\}$, where $S_i = 
  \{s^i_1, \ldots, s^i_{n_i}\}$ for all $i$.
  Write $X = \bigcup S = \{v_1, \ldots, v_m\}$.
  The exact cover problem is to decide whether there exists a subset $S'$ of $S$
  such that every element of $X$ occurs in exactly one of the sets of~$S'$. 

  Consider the document $D$ with root labeled $\top$ and $n$ $\ind$ children,
  with the $i$-th child having, for all $i$, only one child (with edge probability
  $1/2$), which is a $\drm$ node, and which has $n_i$ child nodes labeled $s^i_1,
  \ldots, s^i_{n_i}$. The document $W$ has root labeled $\top$ and $|X|$ child
  nodes labeled $v_1, \ldots, v_m$.

  $W$ is a possible world if and only if there is some subset of the $\drm$
  nodes whose union yields exactly $W$ (without duplicates), so that the
  reduction shows hardness.

  \medskip

  To show hardness of $\poss{\emp}{\muxdrm}$, observe that the previous proof
  can be adapted directly by replacing $\ind$ nodes by $\mux$ nodes, as every
  $\ind$ node has exactly one child.

  \medskip

  Let us last consider $\poss{\emp}{\muxind}$. For this problem, we show a
  reduction from Boolean satisfiability. We use the same notations for the input
  instance as in the proof of Prop.~\ref{prop:posscie}. We additionally
  introduce $n$ node labels $l_1, \ldots, l_n$, with label $l_i$ corresponding to
  clause~$C_i$.

  Consider the document $D$ whose root is labeled $\top$ and has $m$ $\mux$
  child nodes, each of them having two $\ind$ children with edge probability
  $1/2$, the probabilities of all edges of the $\ind$ nodes being also $1/2$.
  For all $i$, the first $\ind$ child of the $i$-th $\mux$ node has one child
  labeled $l_j$ for every clause $C_j$ where $x_i$ occurs; the second one
  has one child labeled $l_j$ for every clause $C_j$ where $\neg x_i$
  occurs. The document $W$ has root labeled $\top$ and $n$ children, the $i$-th
  one having label $l_i$.

  We claim that $W$ is a possible world of $D$ if and only if $F = \bigwedge C_i$ is
  satisfiable. To see why, we consider a one-to-one mapping which associates, to
  any valuation $\nu$ of $F$, the outcomes of the $\mux$ nodes obtained by
  selecting the first child (resp.\ the second child) of the $i$-th $\mux$ node
  if $\nu(x_i) = \true$ (resp.\ $\nu(x_i) = \false$): by construction, the
  labels of the remaining $\ind$ nodes are those of the clauses which are true
  under valuation $\nu$ (possibly occurring multiple times). Hence, if there is
  a valuation $\nu$ satisfying $F$, then, selecting the outcomes of the $\mux$
  nodes in this fashion, we can ensure that the remaining regular nodes are the
  $l_1, \ldots, l_n$, so that $W$ is a possible world of~$D$ as we can choose a
  valuation of the $\ind$ nodes that keeps exactly one occurrence of each label.

  Conversely, if $W$ is a possible world of~$D$, the outcome of the $\mux$ nodes
  in any outcome of~$D$ realizing~$W$ gives a valuation $\nu$ under which $F$ is
  satisfied. Indeed, consider such an outcome and valuation $\nu$, and, for any
  clause $C_j$ of $F$, let us show that $C_j$ is satisfied by $\nu$. Because $W$
  is achieved, some node $n$ labeled $l_j$ must have been kept, and it must be
  the descendant of a $\mux$ node $n'$ (say the $i$-th). Either it is a child of
  $n'$'s first child $n'_1$, or of $n'$'s second child $n'_2$. In the first case,
  this means that $\nu(x_i) = \true$ because $n'_1$ was retained, and $n$ being
  a child of $n'_1$ means that $x_i$ occurs positively in $C_j$, so that $C_j$
  is true under $\nu$. The second case is analogous.
\end{proof}

\section{Explicit matches}
\label{sec:condition}

We now attempt to understand how the overall hardness of $\pposs$ is caused by
the difficulty of finding how the possible world $W$ can be \emph{matched} to $D$.

\begin{definition}
  \label{def:candidate}
  A \emph{candidate match} of $W$ in $D$ is an injective mapping $f$ from the
  nodes of~$W$ to the regular nodes of $D$ such that, if $r$ is the root of $W$
  then $f(r)$ is the root of $D$, and if $n$ is a child of $n'$ in $W$ then
  there is a descending path from $f(n)$ to $f(n')$ going only through
  probabilistic nodes.
\end{definition}

Intuitively, candidate matches are possible ways to generate $W$ from $D$,
ignoring probabilistic annotations, assuming we can keep exactly the regular
nodes of $D$ that are in the image of $f$.
There are exponentially many candidate matches in general, so it is natural
to ask whether $\pposs$ is tractable if all matches are explicitly
provided as input:

\begin{definition}
  Given a class $\prxml^{\mathcal{C}}$ and $o \in \{\bot, \emp, \tot, \top\}$, the $\pposs$ problem with
  \emph{explicit matches}
  $\posse{o}{\mathcal{C}}$ is the same as the
  $\poss{o}{\mathcal{C}}$ problem except that the set of the candidate matches of
  $W$ in $D$ is provided as input (in addition to $D$ and $W$).
\end{definition}

We study the explicit matches variant as a natural generalization of situations
where the ways to match the possible world~$W$ to the document~$D$ are not too
numerous and can be computed efficiently. For instance, if we assume that node
labels in $W$ are unique, so that there is no ambiguity about how to match~$W$
to~$D$, then we are within the scope of the explicit matches variant, as the
(unique) candidate match can be computed in polynomial time. The same applies to
the situation where we only assume that no two sibling nodes carry the same
label, or to more general settings where the possible matches can be identified
easily. Requiring the possible matches to be provided as input is just a way to
formalize that we are not accounting for the complexity of locating those
matches.

We first note that explicit matches ensure tractability of all local
dependency models, by reduction to deterministic tree
automata~\cite{cohen2009running}, this time also for unordered
documents. Intuitively, we can consider all candidate matches separately and
compute the probability of each one, in which case
no label ambiguity remains so any order can be imposed:

\begin{theorem}
  \label{thm:posscemuxinddrm}
  $\possce{\top}{\muxinddrm}$ can be solved in polynomial time.
\end{theorem}

\begin{proof}
  We prove the theorem using the results of~\cite{cohen2009running}. An
  alternative, stand-alone proof is given in
  Appendix~\ref{apx:posscemuxinddrmproof}.

  We say that a candidate match $f$ is realized if we are in \emph{the} possible
  world where the regular nodes of $D$ that are kept are exactly those of the
  image of $f$. Hence, we can compute the probability of $W$ by summing the
  probability of every candidate match being realized (because these events are
  mutually exclusive).
  
  Now, to compute the probability of a candidate match~$f$, replace the labels
  of nodes of $W$ by unique labels (yielding $W'$) and
  replace the labels of every node $n$ of $D$ in the image of $f$ by the label
  of $f^{-1}(n)$ in $W'$, to obtain a probabilistic document~$D'$. The
  probability of $f$ being realized is $D'(W')$. Importantly, if $D$ and $W$ are
  unordered, we can make $D'$ and $W'$ ordered by choosing any order on sibling
  nodes in $D'$, and apply the same order (following $f^{-1}$) to sibling nodes
  in $W$; this works because the way to match $W$ to $D$ is fully specified by
  $f$ so there is no matching ambiguity when imposing this order.

  This concludes the proof, because $D'$ and $W'$ are computable in polynomial
  time and $D'(W')$ can be computed by a deterministic tree automaton as in the
  proof of Theorem~\ref{thm:possctotmuxinddrm}.
\end{proof}

For long-distance dependencies, however, it is easily seen that $\pposs$ is
still hard with conjunction of events, even if explicit matches are provided:

\begin{theorem}
  \label{thm:possecie}
  $\posse{\bot}{\cie}$ is NP-complete, even when $D$ has height $3$.
\end{theorem}

\begin{proof}
  From the proof of Prop.~\ref{prop:posscie}, noticing that there is only
  one (trivial) match of $W$ in~$D$ for the instances considered in the
  reduction.
\end{proof}

This being said, it turns out that the hardness is really caused by event 
\emph{conjunctions}. To see this, we introduce the $\prxml^{\mie}$ class,
which allows only \emph{individual} events:

\begin{definition}
  The $\prxml^{\mie}$ class features \emph{multivalued independent events}
  taking their values from a finite set $V$ (beyond $\true$ and
  $\false$, with probabilities summing to~$1$), and 
  probabilistic $\mie$ nodes whose child edges are annotated by a \emph{single} event $e$ and
  a value $x \in V$. A $\mie$ node cannot be the child of a $\mie$ node. When
  evaluating $D$ under a valuation $\nu$, child edges of $\mie$ nodes labeled
  $(e, x)$ should be kept if and only if $\nu(e) = x$.
\end{definition}

Note that $\mie$ hierarchies are forbidden (because they can straightforwardly
encode conjunctions), so that $\prxml^{\mie}$ does not capture $\ind$
hierarchies. However, as we introduced it with multivalued (not just Boolean)
events, it captures $\prxml^{\mux}$:

\begin{proposition}
  We can rewrite $\prxml^{\mux}$ to $\prxml^{\mie}$ and
  $\prxml^{\mie}$ to $\prxml^{\cie}$ in PTIME.
\end{proposition}

\begin{proof}
  To rewrite $\prxml^{\mux}$ to $\prxml^{\mie}$, first rewrite the input
  $\prxml^{\mux}$ document to a $\prxml^{\mux}$ document with no $\mux$
  hierarchies (no $\mux$ node is a child of a $\mux$ node); this can be done in
  polynomial time (\cite{abiteboul2009expressiveness}, Lemma~5.1). Next,
  introduce one event per $\mux$ node and one outcome for this event per child
  of the $\mux$, with one additional outcome (to make the probabilities sum
  to~$1$) if the original probabilities of the $\mux$ child edges summed
  to~$<1$. Replace each $\mux$ node by a $\mie$ node, where every child edge of
  the $\mie$ node is labeled by the event introduced for this $\mux$ node and
  the value introduced for the outcome where this child edge is kept. The
  absence of $\mux$ hierarchies ensures that the requirement on the absence of
  $\mie$ hierarchies is respected.

  To rewrite $\prxml^{\mie}$ to $\prxml^{\cie}$, we claim that every
  multivalued event $e$ with $k$ outcomes can be replaced by a set $S_e$ of $O(k)$ Boolean
  events such that each outcome $e = x_i$ can be represented by a conjunction of
  $O(\log_2 k)$ events of $S_e$, those conjunctions having the
  same probability as their original outcome and forming a partition of all
  outcomes of events in $S_e$. Assuming that this claim holds, the
  $\prxml^{\mie}$ document can be rewritten in polynomial time to $\prxml^{\cie}$ by
  performing this encoding for all multivalued events, and replacing every
  $\mie$ node by a $\cie$ node and replacing each child edge labeled $(e, x_i)$
  by a child edge labeled with the corresponding conjunction.

  Now, to see why the claim is true, given a multivalued event $e$, observe that
  we can build a binary decision tree $T_e$ of the outcomes of $e$. Hence, we
  can introduce one Boolean event per internal node of $T_e$, and choose its
  probability according to that of its two child edges in~$T_e$ (the probability of an edge $a$ in
  $T_e$ being the total probability of the outcomes reachable from the target
  of~$a$, normalized by that of the outcomes reachable from the origin of~$a$).
  Hence, we associate to each outcome $x_i$ of $e$ the conjunction of Boolean
  choices leading to $x_i$ in $T_e$: it has the right probability by
  construction, and, for every valuation of the Boolean events, exactly one
  conjunction is true (the one corresponding to the leaf of $T_e$ selected by
  following those choices). Now, as $T_e$ is a binary tree with $k$~leaves (the
  number of outcomes of $e$), it has $O(k)$ internal nodes and its height is
  $O(\log_2 k)$, which proves the
  claim and completes the proof.

  Observe that $\prxml^{\mie}$ does \emph{not} capture $\prxml^{\muxdrm}$; a
  proof of this fact is given in Appendix~\ref{apx:rewriteproof}.
\end{proof}

In the $\prxml^{\mie}$ class, the $\pposs$ problem is still NP-hard, by
reduction to exact cover; however, with explicit matches, the $\ppossc$
problem is tractable, both in the ordered and unordered setting, despite the
long-distance dependencies. Intuitively, the candidate matches are mutually
exclusive, and each match's probability can be computed as that of a conjunction
of equalities and inequalities on the events at the frontier.

\begin{theorem}
  \label{thm:possmie}
  $\poss{\bot}{\mie}$ is NP-complete, even when $D$ has height $3$ and events
  are Boolean.
\end{theorem}

\begin{proof}
  From Prop.~\ref{prop:uposs-upossc} it suffices to show hardness.
  We show a reduction from exact cover, as in the proof of
  Theorem~\ref{thm:possempmuxinddrm}, with the same notation for the exact cover
  instance (and, intuitively, using for~$D$ and~$W$ the straightforward encoding
  to $\prxml^{\mie}$ of the instances used in this last proof to show hardness
  of $\poss{\bot}{\muxdrm}$ and $\poss{\bot}{\inddrm}$).

  Consider a set of $n$ Boolean events $E = \{e_1, \ldots, e_n\}$ (with values
  in $\{\true, \false\}$ and probabilities $1/2$.
  Consider the document $W$ with one root labeled $\top$ and $m$ children
  labeled $l_1, \ldots, l_m$. Consider the $\prxml^{\mie}$ document~$D$ with one
  root labeled~$\top$ and one $\mie$ child with, for $1 \leq j \leq n$, $n_i$
  child edges labeled~$(e_i, \true)$ leading to children labeled $s^i_1, \ldots,
  s^i_{n_i}$. Order in the input $D$ the child nodes of the root in $W$ from
  $l_1$ to $l_m$, and the child nodes of the root in $D$ from those labeled
  $l_1$ to those labeled $l_m$, the order between those carrying the same labels
  being arbitrary, so that we are showing a reduction either to
  $\poss{\emp}{\mie}$ or to $\poss{\tot}{\mie}$. 

  Now, $W$ is a possible world of $D$ if and only if there is a valuation of
  the events of~$E$ such that, for every $1 \leq j \leq m$, there is exactly one
  node labeled $l_j$ that is retained. This amounts to choosing a subset
  $S'$ of $S$ such that every item of $X$ occurs exactly once in~$\bigcup S'$:
  the set $S'$ corresponds to the set of events of $E$ that are evaluated
  to~$\true$. Hence, $(D, W)$ is a positive instance of $\poss{\emp}{\mie}$ if
  and only if $F$ is satisfiable, so that $\poss{\bot}{\mie}$ is NP-hard.
\end{proof}

\begin{theorem}
  \label{thm:posscemie}
  $\possce{\top}{\mie}$ can be solved in polynomial time.
\end{theorem}

\begin{proof}
  Observe first that, as in the proof of Theorem~\ref{thm:posscemuxinddrm}, 
  the probability that $W$ is realized is that of either of the candidate matches
  being realized, those events being mutually exclusive.
  We assume that, if $W$ and $D$ are ordered, we have checked (in PTIME) that
  candidate matches respect the order (for a candidate match $f$, if $v$ and
  $v'$ are sibling nodes in $W$ such that $v$ comes before $v'$, then $f(v)$
  comes before $f(v')$ in the document order of $D$), and removed those which do
  not.

  Now, consider a candidate match $f$. We must compute the probability $p_f$ that $f$
  is realised, namely, that we are in the possible world where the only regular nodes
  that are kept in $D$ are those of the image $I$ of $f$; we abuse notation so
  that we consider $\mie$ nodes of $D$ to be in $I$ if one of their children is
  in $I$. We will write this probability $p_f$ as that of a conjunction of events:
  the events that all nodes in $I$ are kept, and the events that all nodes not
  in $I$ are discarded.

  The event of all nodes in $I$ being kept can be written as the conjunction
  $c_+$ of all $e_i = x_i$ for every edge $(e_i, x_i)$ between a $\mie$ node
  in~$I$ and a child node also in $I$. Indeed, to keep~$I$, all the conditions
  on edges leading to a node of~$I$ must be respected.

  The event of all nodes not in~$I$ being discarded can be written as a
  conjunction $c_-$ of the same kind, in the following fashion. Consider every
  topmost node $n$ not in $I$. If $n$'s parent~$n'$ is a regular node, then the
  overall probability of the match $f$ is $p = 0$, because if we keep $n'$ then
  we must keep $n$; in this case, we can forget about $f$ altogether. Otherwise,
  we add to $c_-$ the atom $e_i \neq x_i$, where $(e_i, x_i)$ is the label of the
  edge from $n'$ to $n$.

  We now have either eliminated $f$ or obtained (in polynomial time) the conjunction $c =
  c_+ \wedge c_-$ which is necessary and sufficient for the match to hold, the
  atoms of $c$ being of the form $e_i = x_i$ or $e_i \neq x_i$, where the
  $e_i$'s are events and the $x_i$'s are outcomes. Now, we can compute in
  polynomial time
  the probability $p_f$ of $c$. Indeed, regroup the atoms by the probabilistic
  event occurring in them. For each probabilistic event $e$, we consider the (possibly
  empty) subset of outcomes satisfying the atoms for $e$, and compute its total
  probability $p_f^e$. As the choices are independent between events, the overall
  probability $p_f$ of $c$ is the product of the $p_f^e$ over all events~$e$.
\end{proof}

\section{Conclusion}
\label{sec:conclusion}

We have characterized the complexity of the counting and decision variants of
$\pposs$
for unordered or
ordered XML documents, and various $\prxml$ classes.
With explicit matches, $\ppossc$ is tractable unless event conjunctions are
allowed. Without explicit matches, $\pposs$ is hard unless dependencies are
local; in this case, if the documents are ordered, $\ppossc$ is tractable,
otherwise $\ppossc$ is hard and $\pposs$ is tractable only with $\ind$ \emph{or}
$\mux$ nodes (and hard if both types, or $\drm$ nodes, are allowed). Our results
are summarized in Table~\ref{tab:results} on page~\pageref{tab:results}.

We note that,
using our results and via translations between the
probabilistic relational and XML models~\cite{amarilli2013connections}, we can
derive some bounds on the complexity of $\pposs$ for relational databases. In terms
of tractability for the (unordered) relational model, we can deduce the
tractability of the decision formulation of
$\pposs$ for the tuple-independent
model~\cite{LakshmananLRS97,DalviS07} and the block-independent-disjoint
model~\cite{barbaraetal92,re2007materialized}, and the tractability of both the
decision and counting variants on
pc-tables~\cite{green2006models,HuangAKO09} under the assumption that explicit matches are
provided and that tuples are annotated by a single equality constraint on a
multivalued event, in the spirit of $\mie$. We remark, however, that such
results are not hard to prove directly in the relational model. In terms of
intractability, we observe that the translation from XML to relational models
in~\cite{amarilli2013connections} requires the introduction of explicit node IDs
for all nodes of the document, so that this does not translate to a reduction
for the $\pposs$ problem: intuitively, the translation of~$W$ to a relational
table would have to specify the exact node IDs to be matched. We leave as
future work a more complete investigation of~$\pposs$ in the relational
context, or the study of possible alternative translations that provide more
reductions for $\pposs$ from one setting to the other.

Additional directions for future work would be to study more precisely the effect of $\drm$ nodes and $\ind$
hierarchies, for instance by attempting to extend the $\prxml^{\mie}$ class to
capture them, or try to understand whether there is a connection between the
algorithms of~\cite{cohen2009running} and the proof of
Thm.~\ref{thm:possmuxind}. It would also be interesting to determine under which
conditions (beyond unique labels) can candidate matches be enumerated in
polynomial time, so that the $\pposs$ problem reduces to the explicit matches
variant. Last but not least, another natural problem setting is to allow the order on
sibling nodes of $D$ to be partly specified. This question is already covered
in~\cite{cohen2009running}, but only when all of the possible orderings are
explicitly enumerated:
investigating the tractability of $\pposs$ for more compact representations, such as partial
orders, is an intriguing problem.

\vspace{-0.9em}

\paragraph{Acknowledgements.} The author thanks Pierre Senellart for careful
proofreading,
useful suggestions, and insightful feedback, the anonymous referees of AMW~2014
and BDA~2014 for their valuable
comments, and M. Lamine Ba and Tang Ruiming for helpful early discussion. This
work has been partly funded by the French government under the X-Data
project and by the French ANR under the NormAtis project.

\bibliographystyle{abbrv}
\bibliography{main}

\appendix\section{Stand-alone proof of Theorem~\ref{thm:possctotmuxinddrm}}
\label{apx:possctotmuxinddrmproof}
We prove the claim by representing ordered trees as words (essentially
following a SAX traversal). First, encode $W$ to a word $e_W$ by such a traversal,
internal nodes with label $a$ being encoded as $a^l C a^r$ where $C$ is the
sequence of the encodings of the children of the node (following their order),
and leaves with label $a$ being encoded as $a^l a^r$.

We now convert $D$ to a weighted non-deterministic automaton $A_D$ (on
words) with
$\epsilon$-transitions; importantly, this automaton is acyclic. We proceed in
the following way. Encode a regular node $n$ with label $a$ as the following
structure: the initial state~$q_i$, the encoding
of the children $(c_i)$ of $n$ in order (the final state of each one being connected to
the initial state of the next one by an $\epsilon$-transition of weight $1$),
the final state $q_f$, and an edge labeled~$a^l$ with probability $1$ from $q_i$
to the initial state of the encoding of $c_1$ (if it exists, otherwise to some
intermediate state~$q$)
and an edge labeled~$a^r$ with probability~$1$ from the final state of the
encoding of the last child (if it exists, otherwise from~$q$) to~$q_f$.

Encode the $\drm$ nodes in the same way except that the two last edges are
labeled by~$\epsilon$ (instead of~$a^l$ and~$a^r$). Encode an $\ind$ node~$n$ like a regular node except that
edges leading to the initial state of the encoding of a child of $n$ are given a
probability~$p$ (the probability of this child) and we add an additional edge
with label~$\epsilon$ and probability~$1-p$ to the same initial state to the
final state of the encoding of that child (corresponding to the choice of not
retaining this child). Encode a $\mux$ node as an initial state $q_i$, an
initial state $q_f$, the encoding of each child in parallel,
$\epsilon$-transitions with probability $1$ from the final state of the encoding
of the children to $q_f$, and $\epsilon$-transitions with adequate probabilities
from~$q_i$ to the initial state of the encoding of each child (or to $q_f$, to
make the probabilities sum to~$1$).

There is a clear correspondence between runs of $A_D$ and possible worlds of
$D$, so that what we have to compute is the probability of the encoding $e_W$ of
$W$ according to~$A_D$.

Now, because $A_D$ is acyclic, it is easy to compute this probability in
polynomial time. Indeed, we can compute dynamically for each state $q$ of $A_D$
and every suffix $s$ of $e_W$ the probability that $s$ is produced by a run from
$q$ to the final state of $A_D$.

The base case is that, at the final state, we produce the empty suffix with
probability~$1$ and any non-empty suffix with probability~$0$.

Now, when considering a non-final state $q$ and suffix $s$, because by
construction the sum of all outgoing transitions of $q$ is $1$, the
probability $p(q, s)$ of producing $s$ from $q$ is computed by summing, for
every
outgoing transition~$a$ starting at state~$q$ (with target state~$q_i$), the probability of the
transition~$a$ multiplied by the following quantity:
either, if $a$ is an $\epsilon$-transition, the value
$p(q_i, s)$ (which was already computed) or, if $a$ has label $x$, either $0$ if
$|s| = 0$ or the first letter of $s$ is not~$x$, or otherwise the value $p(q_i, s')$
(which was already computed) where $s'$ is the suffix of length $|s|-1$ of
$e_W$.

\section{Stand-alone proof of Theorem~\ref{thm:posscemuxinddrm}}
\label{apx:posscemuxinddrmproof}
  
As in the other proof of this result, it suffices to consider a single candidate
match, as the overall probability can be obtained by summing that of each match,
and the decision problem can be solved by considering matches separately. If $D$
and $W$ are ordered, we can assume that matches which do not satisfy the order
constraints have been discarded. For simplicity we relax the restriction that
the probability of edges is always $<1$, so that we can encode $\drm$ nodes as
$\ind$ nodes and consider only $\ind$ and $\mux$ nodes.

\medskip

Let us first prove that $\posse{\top}{\muxinddrm}$ can be solved in polynomial
time.

Consider a candidate match $f$. Because all probabilistic choices are
independent, it is clear that all nodes of $D$ in the image $I$ of $f$ can be
kept if and only if there is no $\mux$ node $n$ such that $n'$ and $n''$ are in
$I$ and $n'$, $n''$ are descendants of two distinct children of $n$. Indeed,
this condition is clearly necessary, and, except for this, all choices are
independent within $I$ so they can all be made to succeed\footnote{Remember that
there are no edges with probability $0$.} so that the nodes of $I$ are retained.

So, assuming that this condition holds (it can be checked in polynomial time), the
question is only to see whether the nodes not in $I$ can be discarded. To check
this, we define recursively on all nodes $n$ of $D$, in a bottom-up fashion, the
Boolean value $e(n)$ indicating if $n$ can be ``empty'', that is, if there is a
possible world rooted at $n$ that is empty.

If $n$ is a regular node then we define $e(n) = \false$.

If $n$ is an $\ind$ node, we define $e(n) = \true$ if and only if $e(n)$ is
$\true$ for all the children of $n$ with edge probability~$1$ (remember
that we relaxed the condition on probabilities being~$<1$ because we encoded
$\drm$ nodes as $\ind$ nodes). In particular, if $n$ has no children with edge
probability~$1$, we define $e(n) = \true$.

If $n$ is a $\mux$ node, we define $e(n) = \true$ if and only if the
probabilities of~$n$ sum up to~$<1$, or there exists a child of~$n$ such that
$e(n)$ is $\true$.

Now, we can use $e$ to express the fact that it should be possible to discard
all regular nodes of $D$ except those in $I$. To do so, by a slight abuse of
terms, we say that a probabilistic node is in $I$ if it has a regular descendant
that is in $I$. Now, we claim that the match can yield $W$ if and only if, for every topmost
node $n$ not in $I$, either $e(n)$ is true, or the parent of $n$ (which by
definition is in $I$) is a $\mux$ node or is an $\ind$ node $n'$ such that the
edge from~$n'$ to~$n$ is labeled with a probability~$<1$. To see why this claim
holds, observe that, if this condition is respected, all nodes not in~$I$ can be
discarded (either because $e(n)$ is $\true$ so we can choose an empty subtree as
the possible world rooted at them, or by deciding to discard them at the level
of their parent -- for $\mux$ nodes, in fact, we have no choice but to discard
them). Conversely, if this condition does not hold for a node $n'$ of~$D$,
then~$n'$ must be kept, and the possible world chosen for the subtree rooted at
$n'$ will have to be non-empty because $e(n')$ is $\false$.

This condition can be tested in polynomial time. Hence,
$\posse{\top}{\muxinddrm}$ can be solved in polynomial time by checking if one
of the matches is acceptable in this sense.

\medskip

Let us now prove that $\possce{\top}{\muxinddrm}$ can be solved in polynomial
time. We assume that candidate matches are filtered (in polynomial time)
according to the process described above, so as to only keep the matches with
probability $>0$.

Now, the probability that the match~$f$ is realized can be computed as the
probability of keeping its image $I$ (including probabilistic nodes like in the
previous proof), times the probability of discarding the other nodes: indeed, as
$I$ is a rooted subtree, we must first decide outcomes of nodes and edges in
this subtree so that $I$~is kept, and then outcomes such that the rest is
discarded.

It is easily seen that the probability $p_+$ that $I$ is kept is the product of
all probabilities that annotate the edges that are between nodes in $I$: all
$\ind$ edges of the match must be kept (and those draws are performed
independently), and the right $\mux$ edges must always have been chosen
(remember that a $\mux$ node $n$ is in $I$ only if it has a descendant in~$I$,
and by the condition that the match probability is~$>0$ all descendants of~$n$
are descendants of the same child of~$n$).

Now, we must compute the probability $p_-$ that the nodes not in $I$ are
discarded. To do so, we define $e(\cdot)$, as in the previous proof, but as a
probability rather than a Boolean value: $e(n)$ is the probability of the empty
subtree among the possible worlds for the subtrees rooted at $e(n)$ (note that
this probability does not depend on the choices performed elsewhere in the
tree). Once again, we compute $e(\cdot)$ bottom-up.

For a regular node $n$, we define $e(n) = 0$.

For a $\mux$ node $n$, we define $e(n) = (\sum_i p_i e(n_i)) + (1 - \sum_i p_i)$
where the $n_i$ are the children of $n$ and the $p_i$ the corresponding edge
labels. Intuitively, the probability of the $\mux$ to be empty is that of its
children being empty, weighted by their probability, plus the probability that
we select no children (when the probabilities sum to $<1$).

For an $\ind$ node $n$, we define $e(n) = \prod_i ((1 - p_i) + p_i e(n_i))$ with the
same notation. Intuitively, the probability of the $\ind$ to be empty is that
of each child subtree being missing or empty, which occurs either when the
corresponding is removed, or when it is kept but the subtree is empty (summing
those two cases are they are mutually exclusive).

Now, all nodes not in $I$ are discarded if and only if, for each topmost node $n$ not in
$I$, either $n$ is dropped (its parent edge is removed) or the possible world
rooted at~$n$ is empty. This is a conjunction of events, and they are
independent once conditioned by the fact that the nodes in $I$ are kept
(so the outcomes of all $\mux$ nodes in $I$ have already been decided), so we
can compute the overall probability of $f$ as $p_+ p_-$, with $p_-$ being the
product of the probability $p_n$, for all topmost nodes $n$ not in $I$, that $n$ is
dropped or the subtree rooted at~$n$ is empty. Consider $n'$ the parent of~$n$:
the probability $p_n$ is $e(n)$
if $n'$ is a regular node (as $n$ cannot be dropped then), is $1$ if $n'$ is a
$\mux$ node (as $n'$ is in $I$, it has a descendant~$n''$ in $I$, which cannot
be a descendant of $n$ as $n$ is not in $I$, so that when deciding to keep~$n''$
we have already decided that $n$ would be dropped), and it is $p e(n) + (1-p)$
if $n'$ is an $\ind$ node and the probability of the edge from $n'$ to $n$ is
$p$.

Hence, the overall probability, $p_+ p_-$, can be computed in polynomial time,
which concludes the proof.

\section{$\prxml^{\mie}$ does not capture $\prxml^{\muxdrm}$}
\label{apx:rewriteproof}

In this section, we show that $\prxml^{\mie}$ is \emph{not} more general than
$\prxml^{\muxdrm}$, namely, there is no PTIME encoding from $\prxml^{\muxdrm}$
documents to $\prxml^{\mie}$ documents.

\medskip

Consider the $\prxml^{\muxdrm}$ document~$D_n$ with root labeled~$\top$ and one
$\mux$ child that has two children (with edge probabilities $1/2$): one regular
child with label~$c$, and one $\det$ child. The $\det$ node has $n$ $\mux$
children: for all~$i$, the $i$-th of them has edge probabilities $1/2$ and two
regular children with labels~$a_i$ and $b_i$. We show that any encoding of~$D_n$
to a $\prxml^{\mie}$ document~$D'_n$ (having the same possible worlds as $D_n$)
must have size exponential in~$n$.

The document $D_n'$ must have root labeled $\top$, and the root clearly cannot
have any regular children; so it must have $\mie$ children, and without loss of
generality it has only one of them. Now, as $\mie$ hierarchies are not
permitted, all children of this $\mie$ node are regular nodes; clearly they
cannot have any regular children, and without loss of generality they have no
(useless) $\mie$ children. So the only thing to define is the label and edge
labels of the children of this unique $\mie$ node. Without loss of generality we
assume that we remove edges labeled with $(e, v)$ where $e=v$ has
probability~$0$. Clearly the node labels can be assumed to be either~$c$
or~$a_i$ or~$b_i$ for some~$i$.

As all possible worlds of~$D'_n$ must contain at most one node labeled~$c$, we
claim that the parent edge of all child nodes with label~$c$ must be labeled
with the same event~$e$. Indeed, if there are two nodes with label~$c$ and with
edge labels~$(e_1, v_1)$ and~$(e_2, v_2)$ with $e_1 \neq e_2$, because $e_1$ and
$e_2$ are independent and (we assumed) the events $e_1 = v_1$ and $e_2 = v_2$
have probability~$>0$, any valuation where $e_1 = v_1$ and $e_2 = v_2$ yields a
possible world with two~$c$ children, a contradiction. Hence all child nodes
with label~$c$ are labeled with the same event~$e$. Note that, as some possible
world of~$D'_n$ must contain a node labeled~$c$, there has to be at least one
child~$n_c$ with label~$c$.

Now, no possible world of~$D'_n$ contains both a child labeled~$a_i$ or~$b_i$
(for any~$i$) and a child labeled~$c$, so we claim that the parent edge of all
child nodes with label~$a_i$ or~$b_i$ (for any~$i$) must be labeled with this
same event~$e$. Indeed, assume that one such node is labeled with some condition
$(e', x)$ with $e' \neq e$; calling $(e, v)$ the edge label of~$n_c$, any
valuation where $e' = x$ and $e = v$ yields a possible world with a~$c$ node and
an~$a_i$ node or a~$b_i$ node for some~$i$, a contradiction. Hence, in fact, all
child nodes of the $\mie$ node are labeled with the same event $e$.

Now, as $D_n$ has $2^{n}+1$ possible worlds, $e$ must have $\Omega(2^{n})$
different possible values. Hence, $D'_n$ is of size exponential in~$n$. This
concludes the proof.

\end{document}